\RequirePackage{amsmath}
\documentclass{article}
\usepackage[a4paper,hmargin=2.8cm,vmargin=3cm]{geometry}
\usepackage[utf8]{inputenc}
\usepackage{amsfonts}
\usepackage{amsmath}
\usepackage{amssymb, bm}
\usepackage{stmaryrd}
\usepackage{amssymb}
\usepackage{amsthm}
\usepackage{mathtools, bm}
\usepackage{algorithm}
\usepackage[noend]{algorithmic}
\usepackage{color}                      
\usepackage{hyperref}
\usepackage{graphicx}
\usepackage{color}

\newcommand{\floor}[1]{\left\lfloor #1\right\rfloor}

\newtheorem{theorem}{Theorem}
\newtheorem{lemma}[theorem]{Lemma}
\newtheorem{corollary}[theorem]{Corollary}
\newtheorem{remark}{Remark}

\newcommand{\ZZ}{\mathbb{Z}}
\newcommand{\RR}{\mathbb{R}}
\DeclareMathOperator{\conv}{conv}

\begin{document}

\title{Efficient Algorithms to Test Digital Convexity}
\author{
    Lo\"ic Crombez\\
    Universit\'e Clermont Auvergne and\\
    LIMOS\\ 
    Clermont-Ferrand, France\\
    lcrombez@isima.fr
    \and
    Guilherme D. da Fonseca\\
    Universit\'e Clermont Auvergne and\\
    LIMOS\\ 
    Clermont-Ferrand, France\\
    fonseca@isima.fr
    \and
    Yan G\'erard\\
    Universit\'e Clermont Auvergne and\\
    LIMOS\\ 
    Clermont-Ferrand, France\\
    yan.gerard@uca.fr
}

\date{}
\maketitle

\begin{abstract}
A set $S \subset \ZZ^d$ is \emph{digital convex} if $\conv(S) \cap \ZZ^d = S$, where $\conv(S)$ denotes the convex hull of $S$.
In this paper, we consider the algorithmic problem of testing whether a given set $S$ of $n$ lattice points is digital convex.
Although convex hull computation requires $\Omega(n \log n)$ time even for dimension $d = 2$, we provide an algorithm for testing the digital convexity of $S\subset \ZZ ^2$ in $O(n + h \log r)$ time, where $h$ is the number of edges of the convex hull and $r$ is the diameter of $S$.
This main result is obtained by proving that if $S$ is digital convex, then the well-known quickhull algorithm computes the convex hull of $S$ in linear time.
In fixed dimension $d$, we present the first polynomial algorithm to test digital convexity, as well as a simpler and more practical algorithm whose running time may not be polynomial in $n$ for certain inputs.
\end{abstract}

\section{Introduction}

Digital geometry is the field of mathematics that studies the geometry of points with integer coordinates, also known as \emph{lattice points}~\cite{KlR04}.
Convexity is a fundamental concept in digital geometry, as well as in continuous geometry~\cite{Ro89}.
From a historical perspective, the study of digital convexity dates back to the works of Minkowski~\cite{Min10} and it is the main subject of the mathematical field of geometry of numbers.

While convexity has a unique well stated definition in any linear space, different definitions have been investigated in $\ZZ ^2$ and $\ZZ ^3$~\cite{KR82,KR82-2,Cha83,Ki96,ChR1998}. In two dimensions, we encounter at least five different approaches, called respectively digital line, triangle, line~\cite{KR82}, HV (for Horizontal and Vertical~\cite{BDNP96}), and Q (for Quadrant~\cite{Da01}) convexities. These definitions were created in order to guarantee that a digital convex set is connected (in terms of the induced grid subgraph), which simplifies several algorithmic problems.

The original definition of digital convexity in the geometry of number does not guarantee connectivity of the grid subgraph, but provides several other important mathematical properties, such as being preserved under certain affine transformations (Fig.~\ref{f:shearing}). The definition is the following. A set of lattice points $S \subset \ZZ^d$ is \emph{digital convex} if $\conv(S) \cap \ZZ^d = S$, where $\conv(S)$ denotes the convex hull of $S$.

\begin{figure}[tb]
\begin{center}
        \includegraphics[width=150px]{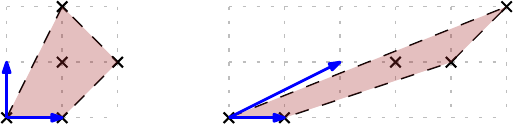}
    \caption{\textbf{Shearing a digital convex set.} Example of a set whose connectivity is lost after a linear shear.} \label{f:shearing}
    \end{center}
\end{figure}


Herein, we consider the fundamental problem of verifying whether a given set of lattice points is digital convex.

\medskip
Problem \texttt{TestConvexity$(S)$}\\
  \textbf{Input: }Set $S \subset \ZZ^d$ of $n$ lattice points given by their coordinates.\\ 
  \textbf{Output: }Determine whether $S$ is digital convex or not.
\medskip
  
The input of \texttt{TestConvexity$(S)$} is an unstructured finite lattice set (without repeating elements). Related work considered more structured data in dimension $2$, in which $S$ is assumed to be connected. The \emph{contour} of a connected set $S$ of lattice points is the ordered list of the points of $S$ having a grid neighbor outside $S$. When $S$ is connected, it is possible to represent $S$ by its contour, either directly as in~\cite{DRR2003} or encoded as binary word~\cite{BLP2009}. The algorithms presented in~\cite{DRR2003,BLP2009} test digital convexity in linear time on the respective input representations.

Our work, however, does not make any assumption on $S$ being connected, or any particular ordering of the input. In this setting, a naive approach to test the digital convexity is: 

\begin{enumerate}
    \item Compute the convex hull $\conv (S)$ of the $n$ lattice points of $S$.
    \item Compute the number $n'$ of lattice points inside the convex hull  of $S$.
    \item If $n=n'$, then $S$ is convex. Otherwise, it is not.
\end{enumerate}


Step 1 consists of computing the convex hull of $n$ points. The field of computational geometry provides a  plethora of algorithms to compute the convex hull of a  finite set $S\subset \RR ^d$ of $n$ points~\cite{BCK08}. The fastest algorithms for dimensions $2$ and $3$ take $O(n \log n)$ time~\cite{Yao81}, which matches the lower bound in the algebraic decision tree model of computation~\cite{PrH77}. 
In dimension $d \leq 3$, if we also take into consideration the output size $h$, i.e. the number of vertices of the convex hull, the fastest algorithms take $O(n \log h)$ time~\cite{KiS86,Cha96}. Some polytopes with $n$ vertices (e.g., the cyclic polytope) have $\Theta(n^{\floor{(d-1)/2}})$ facets. Therefore, any algorithm that outputs this facet description of the convex hull requires $\Omega(n^{\floor{(d-1)/2}})$ time. Optimal algorithms to compute the convex hull in dimension $d \geq 4$ match this lower bound~\cite{ChB93}.

Step 2 consists of computing the number of lattice points inside a polytope (represented by its vertices), which is a well studied problem. In dimension $2$, it can be solved using Pick's formula~\cite{Pic1899}. The question has been widely investigated in the framework of the geometry of numbers, from Ehrhart theory~\cite{Ehr62} to Barvinok's algorithm~\cite{Bar94}. Currently best known algorithms have a complexity of $O(n^{O(d)})$ for fixed dimension $d$~\cite{Bar94-2}. 
As conclusion, the time complexity of this naive approach is at least the one of the computation of the convex hull.

\subsection{Results}

In Section~\ref{s:2d}, we consider the 2-dimensional version of the problem and show that the convex hull of digital convex sets can be computed in linear time.
Our main result is an algorithm for dimension $d=2$ to solve \texttt{TestConvexity$(S)$} in $O(n + h \log r)$ time, where $h$ is the number of edges of the convex hull and $r$ is the diameter of $S$.

In Section~\ref{s:fixedd}, we consider the problem in fixed dimension $d$. We present the first polynomial-time algorithm to test digital convexity, as well as a simpler and more practical algorithm whose running time may not be polynomial in $n$ for certain inputs.

\section{Digital Convexity in 2 Dimensions}
\label{s:2d}

The purpose of this section is to provide an algorithm to test the convexity of a finite lattice $S \subset \ZZ ^2$ in linear time in $n$. To this endeavour, we show that the convex hull of a digital convex set $S$ can be computed in linear time. In fact, we show that this linear running time is achieved by the well-known quickhull algorithm~\cite{BBD96}.

Quickhull is one the many early algorithms to compute the convex hull in dimension $2$. Its worst case time is $O(n^2)$, which makes it generally less attractive than the $O(n \log n)$ algorithm.
However for certain inputs and variations of the algorithm, the average time complexity is reduced to  $O(n \log n)$ or $O(n)$~\cite{BCK08,Gre90}.

The quickhull algorithm starts by initializing a convex polygon in the following manner. First it computes the top-most and bottom-most points of the set. Then it computes the two extreme points in the normal direction of the line supported by the top-most and bottom-most points.
Those four points describe a convex polygon that we call a \emph{partial hull}, which is contained inside the convex hull of $S$. The points contained in the interior of the partial hull are discarded.
Furthermore, horizontal lines and lines parallel to the top-most to bottom-most line passing through these points describe an outlying bounding box in which the convex hull lies (Fig. \ref{qhull}). 

\begin{figure}[tb]
\begin{center}
        \includegraphics[width=150px]{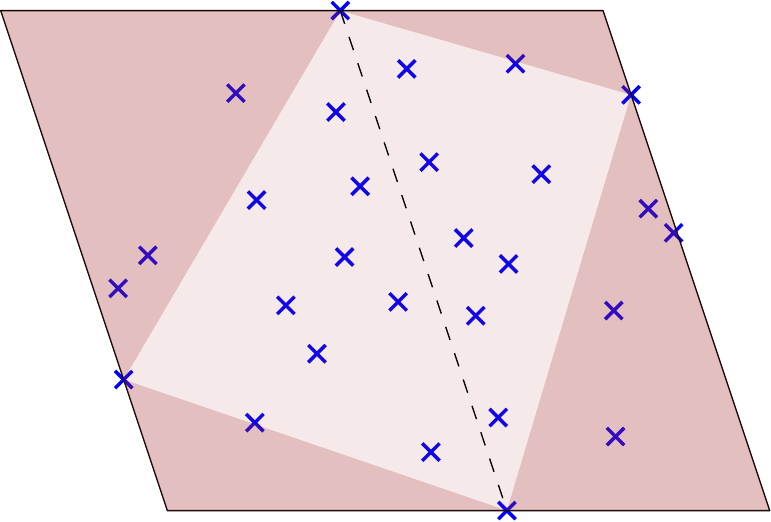}
    \caption{\textbf{Quickhull initialization.} Points inside the partial hull (light brown) are discarded. The remaining points are potentially part of the hull.} \label{qhull}
    \end{center}
\end{figure}

The algorithm adds vertices to the partial hull until it obtains the actual convex hull. This is done by inserting new vertices in the partial hull one by one.
Given an edge of the partial hull, let $v$ denote its outwards normal vector. The algorithm searches for the extreme point in direction $v$. If this point is already an edge point, then the edge is part of the convex hull. Otherwise, we insert the farthest point found between the two edge vertices, discarding the points that are inside the new partial hull. 
Throughout this paper, we call a \emph{step} of the quickhull algorithm the computation of the farthest point of every edge for a given partial hull.
When adding new vertices to the partial hull, the region inside the partial hull expands. Points inside that expansion are discarded by quickhull and herein we name this region \emph{discarded region}. The points that still lie outside the partial hull are preserved, and we call the region within which points might still lies \emph{preserved region} (Fig. \ref{qhull_split}).







We show that quickhull steps takes linear time and that at each step half of the remaining input points of the convex hull is discarded. Therefore, as in standard decimation algorithms, the total running time remains linear. In Section~\ref{ss:22}, we explain how to use this algorithm to test the digital convexity of any lattice set in linear time in $n$.


\begin{theorem}\label{t2}
If the input is a digital convex set of $n$ points, then QuickHull has $O(n)$ time and space complexities.
\end{theorem}


\subsection{Proof of Theorem \ref{t2}}


We prove Theorem \ref{t2} with the help of the following lemma.

\begin{lemma}
The area of the discarded region is larger than the area of the preserved region.
\end{lemma}

%


\begin{proof} 
Consider one step of the algorithm: 
Let $ab$ be the edge associated to the step. When $a$ was added to the hull, it was as the farthest point in a given direction. Hence, there is no point behind the line orthogonal to this direction going through $a$. (Fig. \ref{qhull_split}b). The same can be said for $b$. 
Let $c$ be the intersection point of those two lines. 
Every point that lies within $\triangle abc$ will be fed to the following steps. 
At this step, we are looking for the point that is the farthest from the supporting line of $ab$ and outside the partial hull (let that point be $d$) (Fig. \ref{area_proof}).
Let $e$ and $f$ be the intersections between the line parallel to $ab$ going through $d$, and respectively $ac$ and $bc$. There are no points from $S$ inside the triangle $\triangle cef$. 
Adding $d$ to the partial hull creates two other edges to further be treated: one with $ad$ as an edge that will be fed the points inside $\triangle ade$ and one with $bd$ as the edge that will be fed the points inside $\triangle bdf$. 
The triangle $\triangle abd$ lies within the partial hull, therefore $\triangle abd$ is the region in which points are discarded. (Fig. \ref{area_proof})\\
\begin{figure}[tb]
\begin{center}
        \includegraphics[width=320px]{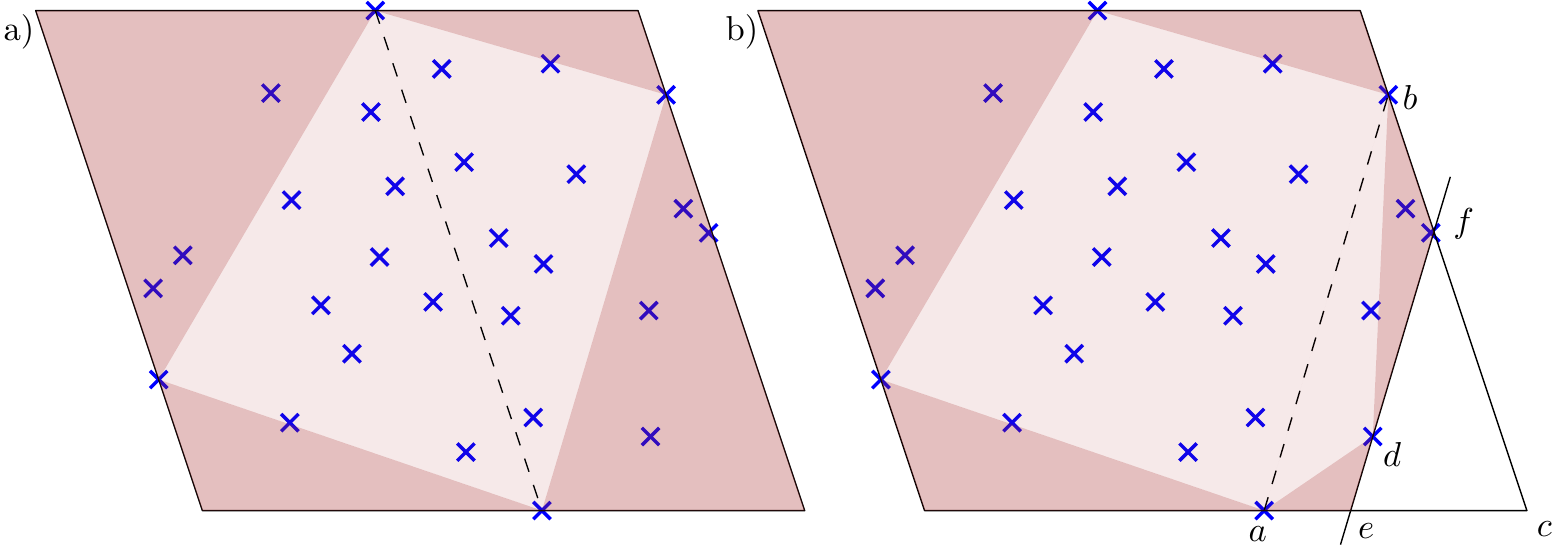}
    \caption{\textbf{Quickhull regions.} The preserved region (region in which we look for the next vertex to be added to the partial hull) is a triangle. This stays true when adding new vertices to the hull (as shown here in the bottom right corner). The partial hull (whose interior is shown in light brown) grows at each vertex insertions to the partial hull. The new region added to the partial hull is called discarded region.} 
    \label{qhull_split}
    \end{center}
\end{figure}
\begin{figure}[ht]
    \begin{center}
        \includegraphics[width=150px]{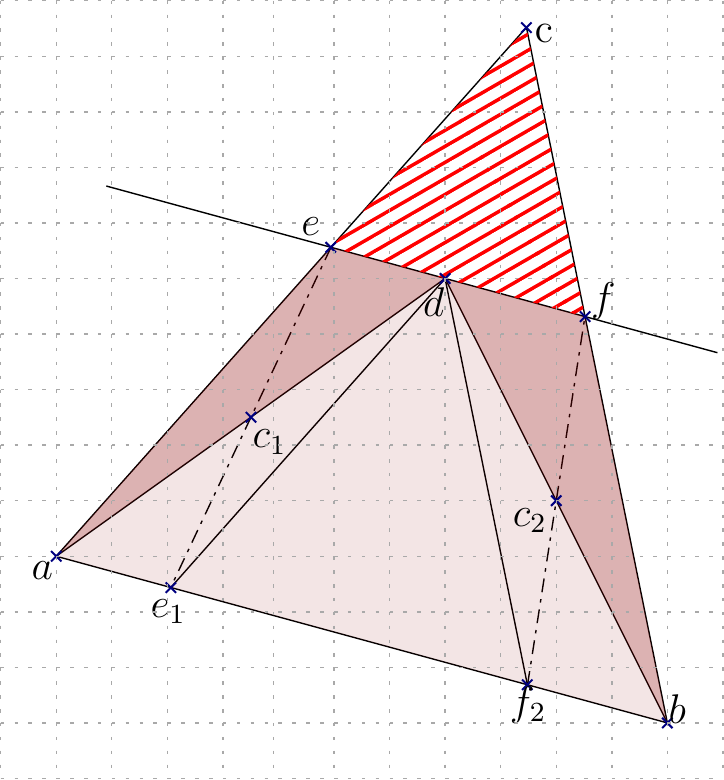}
    \end{center}
    \caption{\textbf{Symmetrical regions.}The next step of the algorithm will only be fed the points inside the dark brown regions (search regions).
    Each lattice points inside the light brown region (discarded region) is inside the partial hull and is therefore discarded.
    Each search region (in dark brown) has a symmetrical region (either through $c_1$ or $c_2$) that lies inside the discarded region. Furthermore, this symmetrical transformation also preserve lattice points. \label{area_proof}}
\end{figure}
We established that the preserved lattice points are the lattice points within $\triangle ade$ and $\triangle bdf$. Also the discarded lattice points are those within $\triangle abd$.
Let $\text{c}_{\text{1}}$ be the middle of $ad$ and $\text{c}_{\text{2}}$ be the middle of $bd$.
As shown in Fig. \ref{area_proof}, the symmetrical of $\triangle ade$ and $\triangle bdf$ through respectively $c_1$ and $c_2$ both lie inside $\triangle abd$ and do not intersect each other. Hence $\triangle abd$ is larger in terms of area than $\triangle aed \cup \triangle bdf$.
\end{proof}

\begin{remark}
Pick's formula does not apply here since all vertices of the triangle (namely $c$ in Fig. \ref{area_proof}) are not necessarily lattice points.
\end{remark}
\begin{remark}
As there is no direct relation between the area of a triangle and the number of lattice points inside it, this result is not sufficient to conclude that a constant proportion of points are discarded at each step.
\end{remark}

\begin{corollary}\label{cor}
The reflection of lattice points inside $\triangle aed$ and $\triangle bdf$ across respectively $\text{c}_{\text{1}}$ and $\text{c}_{\text{2}}$ are lattice points.
\end{corollary}

\begin{proof}
The points $a,b,d$ are lattice points so $\text{c}_{\text{1}}$ and $\text{c}_{\text{2}}$ (middle of respectively $ad$ and $bd$) have their coordinates in multiple of half integers. Hence the reflection of a lattice point across $\text{c}_{\text{1}}$ or $\text{c}_{\text{2}}$ is a lattice point.
Therefore, every lattice point within $\triangle aed$ has a lattice point reflection across $\text{c}_{\text{1}}$ within $\triangle a\text{e}_{\text{1}}d$ and every lattice point within $\triangle bfd$ has a lattice point reflection across $\text{c}_{\text{2}}$ within $\triangle b\text{f}_{\text{2}}d$. 

\end{proof}

\begin{remark}
This previous result would prove that half the points are discarded at each step if it were not for the lattice points on the diagonals $ad$ and $bd$.
\end{remark}



We will now show that quickhull discards at least half of the remaining points at each step, hence proving theorem \ref{t2}

\begin{proof} 
We established in Corollary \ref{cor} that lattice points inside the search regions ($\triangle aed$ and $\triangle bfd$) have symmetrical counterparts inside the discarded region (more precisely inside $\triangle ae_1d$ and $\triangle bf_2d$) (Fig. \ref{area_proof}).
By preserving each points inside $\triangle aed$ and $\triangle bfd$ at each step, we do not have a discarded symmetrical counterpart for the lattice points lying on $ad$ and $bd$. But we do not need to preserve those points, since $ad$ and $bd$ are at this step edges of the partial hull.
Removing lattice points from $ad$ and $bd$ implies that in the following step there will be no lattice points on $ab$, leaving lattice points on $ef$ without a discarded symmetrical counterparts (Fig. \ref{adaptation}).

\begin{figure}[tb] 
    
    \begin{center}
        \includegraphics[width=345px]{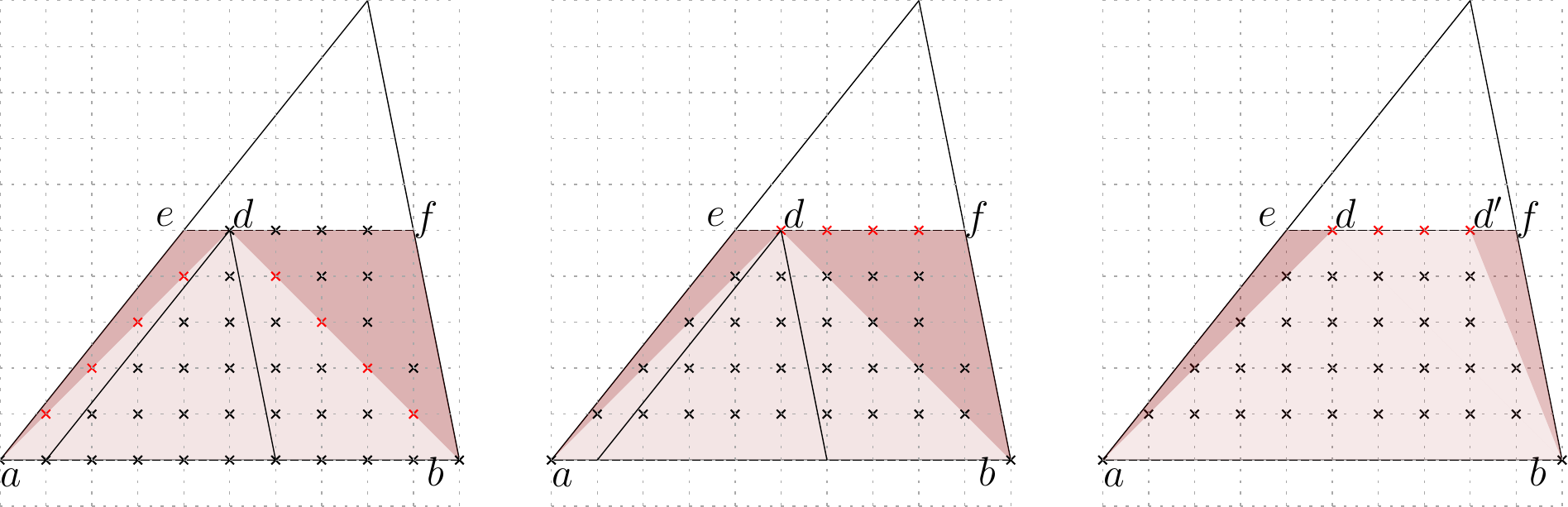}
    \end{center}
    \caption{\textbf{Lonely points.} The lattice points without discarded symmetrical counterparts are shown in red. On the left: if every points inside the triangle is preserved, and in the center: if the points on the edges of the partial hull are discarded. Finally on the right a visualization of what happens if we discard all the farthest points and update the partial hull accordingly.}
    \centering
    \label{adaptation}
\end{figure}

Let actually discard every points on $ef$, since they all are equally farthest from $ab$ in the outer direction, they all belong to the hull. Hence we can add the first and last lattice point on $ef$ to the partial hull (Fig. \ref{adaptation}). Note that this only takes linear time and does not change the time complexity of each individual step.
Hence, at each step of quickhull, for every preserved points there is at least a discarded point. 
Consequently, the number of operations is proportional to  $n\sum\limits_{i=0}^\infty (\frac{1}{2})^i = 2n$ and quickhull takes linear time for digital convex sets.
\end{proof}

\subsection{Determining the digital convexity of a set} \label{ss:22}
We showed in Theorem \ref{t2} that the quickhull algorithm computes the convex hull of digital convex sets in linear time thanks to the fact that at each step quickhull discards at least half of the remaining points. 
By running quickhull on any given set $S$, and stopping the computation if any step of the algorithm discards less than half of the remaining points, we ensure both that the running time is linear, and that if $S$ is digital convex, quickhull finishes and we get the convex hull of $S$. 
If the computation finishes for $S$, we still need to test its digital convexity. To do so, we use the previously computed convex hull and compute $|conv(S) \cap \ZZ^2|$ using Pick's formula~\cite{Pic1899}. The set $S$ is digital convex if $|conv(S) \cap \ZZ^2| = |S|$.
Hence the resulting Algorithm \ref{alg::isconv}.

\begin{algorithm}[h] 
\caption{isDigitalConvex($S$)}
\begin{algorithmic}[1] 
\REQUIRE $S$ a set of points
\ENSURE true if $S$ is digital convex, false if not.
\WHILE{$S$ is not empty}
    \STATE Run one step of the quickhull algorithm on $S$
    \IF{quickhull discarded less than half the remaning points of $S$}
        \RETURN false
    \ENDIF
\ENDWHILE
\STATE Compute $|conv(S) \cap \ZZ^2|$
\IF{$|conv(S) \cap \ZZ^2| > |S|$}
    \RETURN false
\ENDIF
\RETURN true
\end{algorithmic} 
\label{alg::isconv}
\end{algorithm}

\begin{theorem} \label{t::isconv}
Algorithm \ref{alg::isconv} tests digital convexity of any 2 dimensional set $S$, and runs in $O(n + h \log r)$ time, where $h$ is the number of edges of $conv(S)$ and $r$ is the diameter of $S$.
\end{theorem}

\begin{proof}
As Algorithm \ref{alg::isconv} runs quickhull, but stops as soon as less than half the remaining points have been removed, the running time of the quickhull part is bounded by the series $n\sum\limits_{i=0}^\infty (\frac{1}{2})^i = 2n$, and is hence linear. Thanks to Theorem \ref{t2} we know that the computation of quickhull will not stop for any digital convex sets.
Computing $|conv(S) \cap \ZZ^2|$ using Pick's formula requires the computation of the area of $\conv(S)$ and of the number of lattice points lying on its boundary, which requires the computation of a greatest common divisor. Hence this takes $O(h \log r)$ time where $h$ is the number of edge of $conv(S)$ and $r$ is the diameter of $S$.
As $S$ is digital convex if and only if $|S| = |conv(S) \cap \ZZ^2|$, Algorithm \ref{alg::isconv} effectively tests the digital convexity of a 2 dimensional set in  $O(n + h \log r)$ time.
\end{proof}

\section{Test Digital Convexity in Dimension $d$}
\label{s:fixedd}

We provide two algorithms for verifying the digital convexity in any fixed dimension. 

\subsection{Naive algorithm}

The naive algorithm mentioned in the Introduction is based on the following equivalence: the set $S \subset \ZZ ^d$ is digital convex if and only if its cardinality is equal to the cardinality of $conv(S) \cap \ZZ^d$. In Step 1, we compute the convex hull of $S$ (in $O(n \log n + n ^{\lfloor\frac{d}{2}\rfloor})$  time~\cite{ChB93}). In Step 2, we need to count the number of integer points inside $\conv(S)$. The classical algorithm to achieve this goal is known as Barvinok algorithm~\cite{Bar94}.
This approach determines only the number of missing points. If we want to enumerate the points, it is possible to do so through a formal computation of the generating functions used in Barvinok algorithm.

\begin{theorem}
    The naive algorithm tests digital convexity in any fixed dimension $d$ and runs in polynomial time.
\end{theorem}

\begin{proof}
    Computing the convex hull of any set can be done in $O(n \log n + n ^{\lfloor\frac{d}{2}\rfloor})$  time~\cite{ChB93}).
    Counting lattice points inside a convex lattice polytope can be done in polynomial time~\cite{Bar94-2}.
    A direct consequence of the digital convexity definition is that a set $S \subset \ZZ ^d$ is digital convex if and only if $|S| = |conv(S) \cap \ZZ^d|$, hence the naive algorithm tests digital convexity in any fixed dimension $d$ and runs in polynomial time.
    
\end{proof}





\subsection{Alternative algorithm}

This new algorithm computes all integer points in the convex hull of $S$ with a more direct approach. Its principle is to enumerate the points $x$ of a  finite lattice set $S' \subset \ZZ ^d$  surrounding $conv(S) \cap \ZZ ^d$ ($conv(S) \cap \ZZ ^d \subset S' $). In a first variant, we count the number of points of $S'$ belonging to $conv(S)$. At the end, the set $S$ is convex if and only if $|conv (S) \cap \ZZ ^d|$ is equal to the cardinality of $S$. In a second variant, for each point of $S'$, we test whether it belongs to $S$ and in the negative case, we test whether it belongs to the convex hull of $S$. If a point of $S' \setminus S \cap conv(S)$ is found, then $S$ is not convex. 

We define the set $S'$ as the set of points $x\in \ZZ ^d$ such that the cube $x+[-\frac{1}{2} , \frac{1}{2} ]^d $ has a nonempty intersection with the convex hull of $S$, where $+$ denotes the Minkowski sum. It can be easily proved that $S'$ is $2d$-connected (the $2d$ neighbors of a lattice point $x\in \ZZ ^d$ are the $2d$ integer points at Euclidean distance $1$) and by construction, it contains $S$. 
The graph structure induced by the $2d$-connectivity on $S'$ allows to visit all the points of $S'$ efficiently: for each point $x\in S'$, we consider its $2d$ neighbors and test whether they belong to  $S'$. If they do, we add them to the stack of the remaining points of $S'$. The goal is to test whether a point of $S' \setminus S$ is in the convex hull of $S$.
\begin{figure}[tb]
    \begin{center}
        \noindent\makebox[0.9\textwidth]{\includegraphics[width=325px]{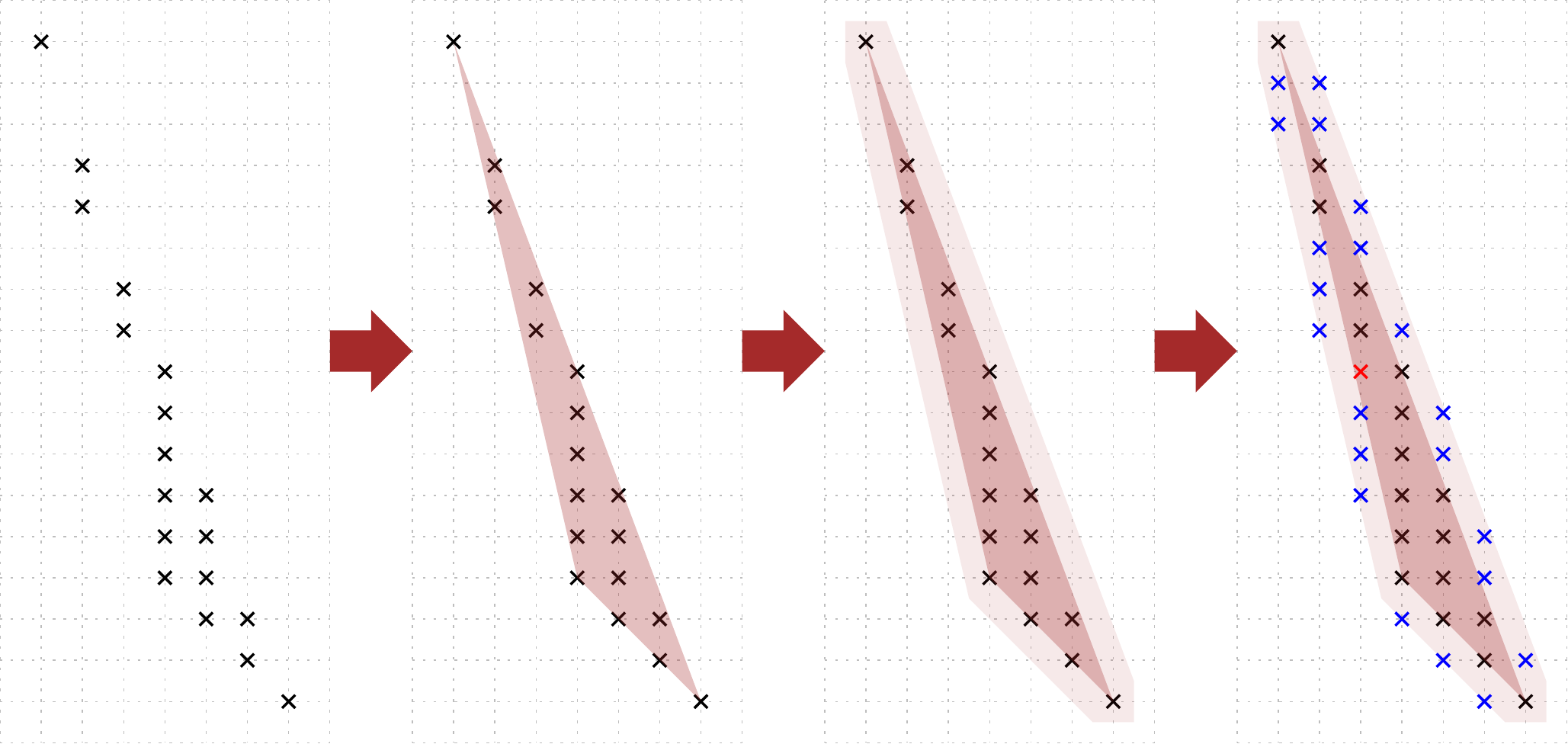}}
    \end{center}            
      \caption{\label{fig7}\textbf{Practical algorithm.} A lattice set $S$, its convex hull and its dilation by a centered cube of side $1$. The intersection of $\conv (S) + [-\frac{1}{2}, \frac{1}{2}] ^d$ with the lattice is the set $S'$. It is $2d$-connected and contains the convex hull of $S$. The principle of the algorithm is either to count the points of $S'$ in $\conv (S)$ (variant 1)  or to search for a point of $S' \setminus S$ (blue points) in the convex hull of $S$ (variant 2). }
    \centering
    \label{qh}
\end{figure}

Then the algorithm has two main routines:
\begin{itemize}
    \item $\mathtt{InConvexHull}_S$ tests whether a given point $x\in \RR  ^d$ belongs to the convex hull of $S$. It is equivalent with testing whether there exists a hyperplane separating $x$ from the points of $S$. It can be done by linear programming with a worst-case time complexity of $O(n)$ for fixed dimension $d$~\cite{BCK08}.
    \item $\mathtt{InConvexHull}_{S+  [-\frac{1}{2}, \frac{1}{2}] ^d }$ tests whether a given point $x$ belongs to the convex hull of $S+  [-\frac{1}{2}, \frac{1}{2}] ^d$. It follows the same principle as 
    $\mathtt{InConvexHull}_S$ with $2^d n$ points. The time complexity remains linear in fixed dimension. This routine is used to test whether an integer point belongs to $S'$.
 \end{itemize}
 
The algorithm is the following. First, we create a stack $T$ of the points of $S'$ to visit and initialize it with the set $S$. For each point $x$ in $T$, we remove it from the stack $T$ and label it as already visited. Then, we consider its $2d$ neighbors $x'$. If $x'$ belongs to $S'$ and has not been visited previously, we add it in the stack $T$. We test whether $x$ belongs to $conv(S)$ and increment the cardinality of $conv (S) \cap \ZZ ^d$ accordingly (variant 1) or test whether $x$ is in $S$ and $ conv(S)$ and return $\mathtt{S}$ \texttt{not convex} if $x \in  conv(S)\setminus S$ (variant 2).  

The running time is strongly dependent on the cardinality of $S'$. It is $ O(n |S'| )$.
If the size of $S'$ is of the same magnitude as the initial set, the algorithm runs in $O(n^2)$ time. 
It is unfortunately not possible to bound $|S' |$ as a function of $n$.
The ratio $\frac{|S' |}{|S|}$ can go to infinity. It is easy to build such an example with a set $S$ consisting of only two lattice points, for instance for any $k \in \mathbb{Z}$ the set $S=\{(0,0);(1,2k)\}$ induces $\frac{|S' |}{|S|} \geq k$. A direction of improvement could be to consider a linear transformation of the lattice $\ZZ ^d$ in order to obtain a more compact lattice set and then a lower ratio $\frac{|S' |}{|S|}$. LLL algorithm~\cite{LLL82} could be useful to achieve this goal in future work. 

As in the naive algorithm, a variant of this approach can be easily developed in order to enumerate the missing points.

\section{Perspectives}


In this paper, we presented an algorithm to test digital convexity in time linear in $n$ for dimension $d=2$. In higher dimensions, our running time depends on the complexity of general convex hull algorithms. The questions of whether digital convexity can be tested in linear time in $3$ dimensions, or faster than convex hull computation in arbitrary dimensions remain open. A tentative approach consists of changing the lattice base, in order to obtain certain connectivity properties.

We showed that the convex hull of a digital convex set in dimension 2 can be computed in linear time. Can the convex hull of digital convex sets be computed in linear time in dimension 3, or more generally, what is the complexity of convex hull computation of a digital convex set in any fixed dimension? We note that the number of faces of any digital convex set in $d$ dimensions is $O(V^{(d-1)/(d+1)})$, where $V$ is the volume of the polytope~\cite{And63,Bar08}. Therefore, the lower bound of $\Omega(n^{\floor{(d-1)/2}})$ for the complexity of the convex hull of arbitrary polytopes does not hold for digital convex sets.


\subsubsection{Acknowledgement}
This work has been sponsored by the French government research program ``Investissements d'Avenir'' through the IDEX-ISITE initiative 16-IDEX-0001 (CAP 20-25).

\bibliographystyle{unsrt}
\bibliography{sample}

\begin{thebibliography}{10}

\bibitem{KlR04}
Reinhard Klette and Azriel Rosenfeld.
\newblock {\em Digital geometry: Geometric methods for digital picture
  analysis}.
\newblock Elsevier, 2004.

\bibitem{Ro89}
Christian Ronse.
\newblock A bibliography on digital and computational convexity (1961-1988).
\newblock {\em IEEE Transactions on Pattern Analysis and Machine Intelligence},
  11(2):181--190, February 1989.

\bibitem{Min10}
H.~Minkowski.
\newblock {\em Geometrie der Zahlen}.
\newblock Number vol.~2 in Geometrie der Zahlen. B.G. Teubner, 1910.

\bibitem{KR82}
Chul~E. Kim and Azriel Rosenfeld.
\newblock Digital straight lines and convexity of digital regions.
\newblock {\em IEEE Transactions on Pattern Analysis and Machine Intelligence},
  4(2):149--153, 1982.

\bibitem{KR82-2}
Chul~E. Kim and Azriel Rosenfeld.
\newblock Convex digital solids.
\newblock {\em {IEEE} Trans. Pattern Anal. Mach. Intell.}, 4(6):612--618, 1982.

\bibitem{Cha83}
Jean-Marc Chassery.
\newblock Discrete convexity: Definition, parametrization, and compatibility
  with continuous convexity.
\newblock {\em Computer Vision, Graphics, and Image Processing}, 21(3):326 --
  344, 1983.

\bibitem{Ki96}
Kazuo Kishimoto.
\newblock Characterizing digital convexity and straightness in terms of length
  and total absolute curvature.
\newblock {\em Computer Vision and Image Understanding}, 63(2):326 -- 333,
  1996.

\bibitem{ChR1998}
Bidyut~Baran Chaudhuri and Azriel Rosenfeld.
\newblock On the computation of the digital convex hull and circular hull of a
  digital region.
\newblock {\em Pattern Recognition}, 31(12):2007 -- 2016, 1998.

\bibitem{BDNP96}
Elena Barcucci, Alberto~Del Lungo, Maurice Nivat, and Renzo Pinzani.
\newblock Reconstructing convex polyominoes from horizontal and vertical
  projections.
\newblock {\em Theoretical Computer Science}, 155(2):321--347, 1996.

\bibitem{Da01}
Alain Daurat.
\newblock Salient points of q-convex sets.
\newblock {\em International Journal of Pattern Recognition and Artificial
  Intelligence}, 15(7):1023--1030, 2001.

\bibitem{DRR2003}
Isabelle Debled-Rennesson, Jean-Luc Rémy, and Jocelyne Rouyer-Degli.
\newblock Detection of the discrete convexity of polyominoes.
\newblock {\em Discrete Applied Mathematics}, 125(1):115 -- 133, 2003.
\newblock 9th International Conference on Discrete Geometry for Computer Im
  agery (DGCI 2000).

\bibitem{BLP2009}
Xavier Provençal Christophe~Reutenauer Srecko~Brlek, Jacques-Olivier~Lachaud.
\newblock Lyndon + christoffel = digitally convex.
\newblock {\em Pattern Recognition}, 42(10):2239 -- 2246, 2009.
\newblock Selected papers from the 14th IAPR International Conference on
  Discrete Geometry for Computer Imagery 2008.

\bibitem{BCK08}
Mark~de Berg, Otfried Cheong, Marc~van Kreveld, and Mark Overmars.
\newblock {\em Computational Geometry: Algorithms and Applications}.
\newblock Springer-Verlag TELOS, Santa Clara, CA, USA, 3rd ed. edition, 2008.

\bibitem{Yao81}
Andrew Chi-Chih Yao.
\newblock A lower bound to finding convex hulls.
\newblock {\em J. ACM}, 28(4):780--787, October 1981.

\bibitem{PrH77}
F.~P. Preparata and S.~J. Hong.
\newblock Convex hulls of finite sets of points in two and three dimensions.
\newblock {\em Communications of the ACM}, 20(2):87--93, February 1977.

\bibitem{KiS86}
D.~Kirkpatrick and R.~Seidel.
\newblock The ultimate planar convex hull algorithm?
\newblock {\em SIAM Journal on Computing}, 15(1):287--299, 1986.

\bibitem{Cha96}
T.~M. Chan.
\newblock Optimal output-sensitive convex hull algorithms in two and three
  dimensions.
\newblock {\em Discrete {\&} Computational Geometry}, 16(4):361--368, Apr 1996.

\bibitem{ChB93}
Bernard Chazelle.
\newblock An optimal convex hull algorithm in any fixed dimension.
\newblock {\em Discrete {\&} Computational Geometry}, 10(4):377--409, Dec 1993.

\bibitem{Pic1899}
Georg Pick.
\newblock Geometrisches zur zahlenlehre.
\newblock {\em Sitzungsberichte des Deutschen
  Naturwissenschaftlich-Medicinischen Vereines für Böhmen "Lotos" in Prag.},
  v.47-48 1899-1900, 1899.

\bibitem{Ehr62}
Eugène Ehrhart.
\newblock Sur les polyèdres rationnels homothétiques à n dimensions.
\newblock Technical report, académie des sciences, Paris, 1962.

\bibitem{Bar94}
Alexander~I. Barvinok.
\newblock A polynomial time algorithm for counting integral points in polyhedra
  when the dimension is fixed.
\newblock {\em Mathematics of Operations Research}, 19(4):769--779, 1994.

\bibitem{Bar94-2}
A.~I. Barvinok.
\newblock Computing the {E}hrhart polynomial of a convex lattice polytope.
\newblock {\em Discrete {\&} Computational Geometry}, 12(1):35--48, Jul 1994.

\bibitem{BBD96}
C.~Bradford Barber, David~P. Dobkin, and Hannu Huhdanpaa.
\newblock The quickhull algorithm for convex hulls.
\newblock {\em ACM Transactions on Mathematical Software}, 22:469--483, 1996.

\bibitem{Gre90}
Jonathan~Scott Greenfield.
\newblock A proof for a quickhull algorithm.
\newblock Technical report, Syracuse University, 1990.

\bibitem{LLL82}
Arjen~K. {Lenstra}, H.~W. {Lenstra}, and L.~{Lovasz}.
\newblock Factoring polynomials with rational coefficients.
\newblock {\em Mathematische Annalen}, 261(4):515--534, 1982.

\bibitem{And63}
G.~E. Andrews.
\newblock A lower bound for the volumes of strictly convex bodies with many
  boundary points.
\newblock {\em Transactions of the American Mathematical Society},
  106:270--279, 1963.

\bibitem{Bar08}
I.~B{\'a}r{\'a}ny.
\newblock Extremal problems for convex lattice polytopes: {A} survey.
\newblock {\em Contemporary Mathematics}, 453:87--103, 2008.

\end{thebibliography}

\end{document}